\documentclass[a4paper, 11pt]{article}
\usepackage{amsmath}
\usepackage{amsthm}
\usepackage{ mathrsfs }
\usepackage{amssymb}
\usepackage{braket}

\newtheorem{definition}{Definition}

\newtheorem{proposition}{Proposition}
\newtheorem{theorem}{Theorem}
\newtheorem{corollary}{Corollary}

\def\be{\begin{equation}}
\def\ee{\end{equation}}
\def\bea{\begin{eqnarray}}
\def\eea{\end{eqnarray}}

\newcommand{\ie}{{\it i.e.}\ }

\newcommand{\R}{\mathbb{R}}
\newcommand{\C}{\mathbb{C}}
\newcommand{\CC}{\mathbb{C}}

\newcommand{\cL}{\mathcal{L}}
\newcommand{\lda}{\lambda}

\newcommand{\1}{\mbox{\hspace{.0em}1\hspace{-.2em}I}}
\newcommand{\id}{\mbox{\hspace{.0em}1\hspace{-.2em}I}}

\newcommand{\Lag}{\mathcal{L}}

\newcommand{\cpb}[2]{\{\! | #1, #2| \! \}}

\newcommand{\vol}{dx\wedge dt}
\renewcommand{\^}{\wedge}
\renewcommand{\epsilon}{\varepsilon}

\renewcommand{\tilde}{\widetilde}
\DeclareMathOperator{\omegaone}{\Omega^{(1)}}

\newcommand{\parder}[2]{\frac{\partial #1}{\partial #2}}
\newcommand{\varder}[2]{\frac{\delta #1}{\delta #2}}
\newcommand{\tilpartial}{\tilde \partial}
\newcommand{\ip}[2]{#1\lrcorner #2}

\numberwithin{equation}{section}
\numberwithin{theorem}{section}

\begin{document}

\begin{center}
	\textbf{\Large A connection between the classical r-matrix formalism and covariant Hamiltonian field theory}\\[3ex]
	\large{Vincent Caudrelier, Matteo Stoppato}\\[3ex]
	
	School of Mathematics, University of Leeds, LS2 9JT, UK  \\[3ex]
\end{center}

\vspace{0.5cm}

\centerline{\bf Abstract}  
We bring together aspects of covariant Hamiltonian field theory and of classical integrable field theories in $1+1$ dimensions. Specifically, our main result is to obtain for the first time the classical $r$-matrix structure within a covariant Poisson bracket for the Lax connection, or Lax one form. This exhibits a certain covariant nature of the classical $r$-matrix with respect to the underlying spacetime variables. The main result is established by means of several prototypical examples of integrable field theories, all equipped with a Zakharov-Shabat type Lax pair. Full details are presented for: $a)$ the sine-Gordon model which provides a relativistic example associated to a classical $r$-matrix of trigonometric type; $b)$ the nonlinear Schr\"odinger equation and the (complex) modified Korteweg-de Vries equation which provide two non-relativistic examples associated to the same classical $r$-matrix of rational type, characteristic of the AKNS hierarchy. The appearance of the $r$-matrix in a covariant Poisson bracket is a signature of the integrability of the field theory in a way that puts the independent variables on equal footing. This is in sharp contrast with the single-time Hamiltonian evolution context usually associated to the $r$-matrix formalism.

\section{Introduction}

The geometrization of Hamiltonian dynamical systems led to a beautiful framework for classical mechanics, see e.g. \cite{Arn} for a modern exposition. The development of an analogous framework for classical field theories followed a less straightforward path and still is the object of current studies, see e.g. the recent book \cite{LSV}. One feature of field theories is that there are several independent (spacetime) coordinates on which the fields depend so that, starting from a Lagrangian description, one has to make a choice from the very beginning. Roughly speaking, one can distinguish two main avenues underlying the current state of the art. 

On the one hand, one can favour one particular coordinate (the time) to perform the Legendre transform and develop the analogous geometrization of Hamiltonian mechanics, resulting in an infinite dimensional Hamiltonian formalism. This point of view seems arbitrary, especially if one is interested in Lorentz invariant theories for instance. Nevertheless, it received a large amount of attention, with a boost coming in particular from the theory of classical integrable systems. The latter provided numerous examples of infinite dimensional Hamiltonian and Liouville integrable systems, since the early examples \cite{ZaF,ZM2}. In that area, important developments such as the theory of Poisson-Lie groups \cite{Drin} and the classical $r$-matrix \cite{STS} have led to an infinite-dimensional version of geometric Hamiltonian mechanics. In parallel, the ``algebraization'' of this framework, driven for instance by I.M. Gel'fand, L.A. Dickey and I. Dorfman, led to what is sometimes called formal (algebraic) variational calculus, see e.g. the books \cite{Dickey,Dorfman}. An important motivation for generalising the classical Hamiltonian theory to field theory in this way was the programme of canonical quantization of integrable field theories into integrable quantum field theories. The classical $r$-matrix method proved to be fundamental to achieve this. It gives rise the notion of quantum $R$ matrix and quantum inverse scattering method \cite{Sklyq,FS,FST,FT2}. 

On the other hand, the conceptual disadvantage of picking a special coordinate to perform the Legendre transform emerged already in the early 1900's. The possibility to generalise the Legendre transform to define conjugate momenta associated to each independent variable naturally leads to a generalisation of the standard Hamilton equations called for short covariant Hamiltonian field theory. This observation is at the basis of a theory discovered independently by De Donder and Weyl and now called de Donder-Weyl formalism \cite{dD,W}. Further developments followed and led to the Lepage-Dedecker theory, see \cite{HK} for a more recent exposition of this theory and a comparison with the de Donder-Weyl formalism. Despite being conceptually the same as the traditional Hamiltonian theory (Lagrangian and Hamiltonian pictures are related by a Legendre transformation), its geometrization shows deep differences. In fact, there is not one established theory of what should play the role of the usual symplectic form and associated symplectic geometry, but instead a variety of related approaches ($k$-symplectic, polysymplectic or multisymplectic) as described in \cite{LSV}. Similarly, the familiar notion of phase space must be promoted to a covariant phase space whose definition and use come with certain difficulties. Such a successful framework is credited to Kijowski and Szczyrba \cite{KijS} and later on Zuckerman \cite{Zuck}. The relation between multisymplectic formalism and the covariant phase space is investigated in \cite{ForgerRomero} and also \cite{Hel} which contains an excellent review of the historical development of the many facets of this field and an account of covariant canonical quantization for free field theories. Alongside the problem of generalising symplectic geometry and the phase space comes the question of generalising to the field theoretic context the variational complex that one can associate to a (Lagrangian) system of (ordinary) equations in mechanics. The relevant structure is the variational bicomplex \cite{And}, see e.g. \cite{Vit} for a review and a guide to the relevant literature and also \cite{Reyes} for the relation between covariant phase space and variational bicomplex. A rigorous approach to the covariant phase space in the framework of jet spaces and Vinogradov secondary calculus was proposed in \cite{Vita}. 

To the best of our knowledge, these two avenues flourished rather independently, driven by motivations with little or no overlap, with the exception of one author, L.A. Dickey, who initiated the investigation of the second, covariant, point of view within the formalism of integrable systems in \cite{D1}. This was further developed in the book \cite{Dickey} where the aforementioned formal algebraic variational calculus was used to describe such objects as multisymplectic forms and the variational bicomplex. Dickey's goal was to study integrable hierarchies from the covariant Hamiltonian point of view, thus breaking the long tradition of the infinite dimensional Hamiltonian formalism that was used in that area, as already mentioned. This body of work does not seem to have been followed up, despite its importance as we now argue. One of the motivations for the endeavour in the aforementioned geometrization of field theory is the programme of covariant canonical quantization as an alternative that would combine the advantages of manifest covariance (as in Feynman's path integral techniques) and ``simple'' quantization rules (as in canonical quantization) without their disadvantages. Our point of view is that integrable field theories are the ``nicest'' field theories one can work with, beyond free field theories, to test the framework. In essence, the quantum inverse scattering method is the manifestation that these theories can be quantized canonically without being plagued by some of the common problems of other theories, such as the need for regularisation and renormalisation procedures.

However, as mentioned above, the quantization procedure of integrable classical theories relies entirely on the classical $r$-matrix which was never considered in Dickey's work and, more importantly, which was thought to be a purely ``single-time'', non covariant, object of the traditional Hamiltonian formalism. The main reason for this belief comes from the way the $r$-matrix appears when formulating an integrable partial differential equation with a Lax pair of the Zakharov-Shabat form \cite{ZS} for instance. The emphasis is on one of the two Lax matrices containing the Cauchy data on which the Poisson bracket is formulated. Crucially, in this approach, the second Lax matrix, describing the time evolution, is a {\it by-product of the first Lax matrix and the classical $r$-matrix}, see e.g. Part I, Chap. III of \cite{FT}. It appears naturally when the Hamilton equations of motion are reinterpreted as a zero curvature equation and an explicit formula for it exists in terms of the classical $r$-matrix and the so-called monodromy matrix of the first Lax matrix (Semenov-Tian-Shansky formula). This is the infinite dimensional analog of an important result of Semenov-Tian-Shansky \cite{STS}. From this point of view, it does not make sense to inquire about the $r$-matrix structure of the second, time Lax matrix. 

However, recently the possible covariant nature of the $r$-matrix has emerged, originally motivated by the specific topic of integrable defects in classical field theories \cite{CK,C}, then followed up by more systematic studies in integrable hierarchies \cite{ACDK,AC}. With the present paper, we show for the first time that indeed the classical $r$-matrix has a covariant nature in the sense that it appears in a covariant Poisson bracket that takes as arguments the Lax form, see our main result Theorem \ref{main_th} below. This answers the question posed by the results in \cite{ACDK,AC} where the same $r$-matrix structure (up to a sign) appears for the space Lax matrix with respect to the traditional Poisson structure and for the time Lax matrix with respect to a new, dual Poisson structure obtained following the old idea of covariant Legendre transformation mentioned above. Since the work of Dickey, this represents the first new step towards a covariant Hamiltonian description of (hierarchies of) integrable classical field theories. We hope that this will revive this topic and allow for a deeper connection with the covariant quantization programme mentioned earlier which could be tested on integrable field theories and, perhaps, related to the quantum $R$ matrix approach.

We note that our main results are systematically obtained from several prototypical examples whose Lagrangians are given explicitly. Therefore, we do not aim at the full generality of a coordinate independent formulation but work directly with the dependent and independent coordinates dictated by our Lagrangians. We display the details of the calculations for each example, as our main goal is to make the results of our paper as accessible as possible despite the fact that it draws on tools from different areas. Our examples are chosen to be sufficiently prototypical to reveal the important main features of our results.

The paper is organised as follows. In Section \ref{generalities}, we present basic elements of the two theories we aim to combine together \ie the covariant Hamiltonian formalism and the classical $r$-matrix structure. The specific choices that we make in selecting the tools we need and the way we use them are dictated by our main goal: obtain a covariant Poisson bracket in which the classical $r$-matrix naturally appears. This is then explained in Section \ref{main_result} where the main results obtained on each example are compiled. As an important by-product, we obtain that the two single-time Poisson bracket structures of the two Lax matrices are indeed governed by the same $r$-matrix, up to a sign. This was the main observation of \cite{CK,C,ACDK} and the original motivation for the present work. Another by-product is the interpretation of the zero curvature equation, or Maurer-Cartan equation for the Lax connection, as a covariant Hamilton equation. In Section \ref{Examples}, our three main examples are studied: the sine-Gordon model, the nonlinear Schr\"odinger equation and the modified Korteweg-de Vries equation. For each one, we follow the structure of Section \ref{main_result} and prove the results in detail. The last section contains some conclusions and comments on certain open problems.

\section{Generalities and results}\label{generalities}

As we aim to combine classical $r$-matrix structures with some aspects of covariant field theory, we now spend some time reviewing those aspects of each framework which will be useful for our purposes. This also serves to introduce our notations and the point of view we take.
In the book \cite{Dickey} by Dickey, an algebraic approach to the variational bicomplex is developed, following the tradition of ``algebraization'' of Lagrangian and Hamiltonian formalism in field theories mentioned in the introduction. We found this approach the most efficient for our purpose and will therefore follow it in this paper. The starting point is the assumption that our models of interest have a Lagrangian description, given in certain chosen independent and dependent coordinates (spacetime variables and fields). Therefore, we will work with local coordinates dictated by the Lagrangian. 

Our strategy to define the covariant Poisson bracket that we need is the following. In \cite{K}, Kanatchikov proposed a definition of covariant Poisson bracket that mimics the one in classical mechanics, {\bf under the assumption that the Lagrangian is first order and non-degenerate}. He used a {\it vertical} differential, to be distinguished from the traditional {\it horizontal} differential $d$ (see below). In \cite{Dickey}, Dickey provides a systematic way to derive a multisymplectic form {\bf starting from a broader class of Lagrangian densities}. Since all known integrable field theories in $1+1$ dimensions possess a Lagrangian formulation of the latter class (certainly the ones we will deal with explicitly in this paper), we can therefore start from such a Lagrangian density, follow Dickey's construction to obtain the corresponding multisymplectic form and then Kanatchikov's construction. Note that we adapt the latter to our needs in the context of $1+1$ dimensional integrable field theories, meaning in particular that we will only consider zero and one (horizontal) forms, see below.

\subsection{Elements of variational calculus with the variational bicomplex}

Let $M=\R^n$ be the \emph{space-time} manifold with coordinates $x^i$, $i=1, \ldots, n$, endowed with a volume form $\omega = dx^1 \wedge \ldots \wedge dx^n$. In the present paper, \textbf{we take $\mathbf n=2$} with $(x^1,x^2) = (x,t)$. The manifold $M$ is viewed as the base manifold in a fibered manifold, as formalised in the the variational bicomplex, see e.g. \cite{Vit}. The typical fibre has local coordinates that represent the fields of the model. One introduces vertical and horizontal differentials $\delta$ and $d$ which satisfies
\be
d^2=0=\delta^2\,,~~d\delta=-\delta d\,,
\ee
so that the operator $d+\delta$ satisfies $(d+\delta)^2=0$.

We now follow \cite{Dickey} in the special case of $n=2$ independent variables which we denote $x$ and $t$ here. 
For convenience, we will only consider theories whose Lagrangian do not depend explicitly on those. Let $\mathcal{K}=\R$ or $\C$. Consider the differential algebra with two commuting derivations $\partial_\mu$, $\mu=1,2$ generated by the commuting variables $u_k^{(\mu)}$, $k=1, \dots, N$, $(\mu)=(\mu_1, \mu_2) $ being a multi-index, quotiented by the relations 
\be
\partial_\nu u_k^{(\mu)}=u_k^{(\mu)+e_\nu}\,,~~e_1=(1,0)\,,~~e_2=(0,1)\,.
\ee 
We simply denote $u_k^{(0,0)}$ by $u_k$, the fields of the theory which would be the local fibre coordinates mentioned above.
We denote this differential algebra by $\mathcal{A}$. We use latin indices to denote the field species and greek indices to denote the space-time coordinates. We will need the notation 
\be
\partial^{(\mu)}=\partial_1^{\mu_1}\partial_2^{\mu_2}=\partial_x^{\mu_1}\partial_t^{\mu_2}\,.
\ee
We consider the space $\mathcal{A}^{(p,q)}$, $p,q\ge 0$ of formal sums of the following form
\begin{equation}
\label{form_pq}
\omega^{(p,q)} =\sum_{(\mu),(i),(\nu)} f^{(\mu)}_{(i),(\nu)}\delta u_{i_1}^{(\mu_1)} \^ \dots \^ \delta u_{i_p}^{(\mu_p)} \^ dx^{\nu_1} \^ \dots \^ dx^{\nu_q}, \qquad f^{(\mu)}_{(i),(\nu)}\in \mathcal{A}
\end{equation}
which are called $(p,q)$-forms. The differentials $(\delta u_k^{(\mu)},dx^\nu)$ are anticommuting and the exterior product $\wedge$ is as usual. We define the operations $d:\mathcal{A}^{(p,q)}\to \mathcal{A}^{(p,q+1)}$ and $\delta:\mathcal{A}^{(p,q)}\to \mathcal{A}^{(p+1,q)}$ as follows. They are graded derivations
\bea
d(\omega_1^{(p_1,q_1)} \wedge \omega_2^{(p_2,q_2)}) = d\omega_1^{(p_1,q_1)} \wedge \omega_2^{(p_2,q_2)}+(-1)^{p_1+q_1}\omega_1^{(p_1,q_1)} \wedge d\omega_2^{(p_2,q_2)},\\
\delta(\omega_1^{(p_1,q_1)} \wedge \omega_2^{(p_2,q_2)}) = \delta\omega_1^{(p_1,q_1)} \wedge \omega_2^{(p_2,q_2)}+(-1)^{p_1+q_1}\omega_1^{(p_1,q_1)} \wedge \delta\omega_2^{(p_2,q_2)}\,,
\eea
and on the generators, they satisfy
\bea
&&df = \sum \partial_\mu \,f dx^\mu=\sum(\parder{f}{x^\mu} + \parder{f}{u_k^{(\nu)}} \, u_k^{(\nu)+e_\mu})dx^\mu\,, \quad f \in \mathcal{A}\,,\\
&&\delta f = \sum \parder{f}{u_k^{(\mu)}} \, \delta u_k^{(\mu)}\,, \quad f \in \mathcal{A}\,,\\
&&\delta(dx^\nu) = \delta(\delta u_k^{(\mu)}) = d(dx^\nu)=0,\\
\label{ddeltagenerator}&&d(\delta u_k^{(\mu)}) = - \delta d u_k^{(\mu)} = - \sum \delta u_k^{(\mu)+e_\nu} \^ dx^\nu.
\eea
This determines the action of $d$ and $\delta$ on any form as in \eqref{form_pq}. As a consequence, one can show that $d^2=\delta^2=0$ and $d\delta =-\delta d$. For our purpose, it is sufficient to take the following (simplified) definition for the variational bicomplex: it is the space ${\cal A}^*=\bigoplus_{p,q}{\cal A}^{(p,q)}$ equipped with the two derivation $d$ and $\delta$.
Note that the direct sum over $q$ is finite and runs from $0$ (scalars) to $n=2$ (volume horizontal forms) whereas the sum over $p$ runs from $0$ to infinity. Of course, each form in ${\cal A}^*$ only contains a finite sum of elements of the form \eqref{form_pq} for certain values of $p$ and $q$. The bicomplex ${\cal A}^*$ generates an associated complex ${\cal A}^{(r)}=\bigoplus_{p+q=r}{\cal A}^{(p,q)}$ and derivation $d+\delta$. Dual to the notion of forms is the notion of vector fields. We consider the dual space of vector fields ${\cal T}{\cal A}$ to the space of one-forms ${\cal A}^{(1)}$ with elements of the form
\be
\label{vector_field}
\xi=\sum_{k,(\mu)} \xi_{k,(\mu)}\, \partial_{u_{k}^{(\mu)}}+\sum_\nu \xi^*_\nu\, \partial_\nu\,.
\ee
The interior product with a form is obtained in the usual graded way together with the rule
\be
\ip{\partial_\mu}{dx^\nu}=\delta_{\mu\nu}\,,~~\ip{\partial_{u_{k}^{(\mu)}}}{\delta u_j^{(\nu)}}=\delta_{kj}\delta_{\mu_1\nu_1}\delta_{\mu_2\nu_2}\,.
\ee
For instance
\bea
&&\ip{\partial_\mu}{(\delta{u_k^{(\sigma)}}\wedge dx^\mu \wedge dx^\nu)}=-\delta{u_k^{(\sigma)}} \wedge dx^\nu\,,\\
&&\ip{\partial_{u_k^{(\mu)}}}{(\delta{u_l^{(\sigma)}}\wedge \delta{u_k^{(\mu)}} \wedge dx^\nu)}=-\delta{u_l^{(\sigma)}} \wedge dx^\nu\,.
\eea
There exist important results regarding the local and global exactness of the vertical and horizontal sequences in the variational bicomplex \cite{Tul,Vit}.
Of special importance for us is the following proposition (cf. Proposition 19.4.4 in \cite{Dickey}) which is related to the formulation of the variation of an action in this setup.
\begin{proposition}
Let  $F =fdx\wedge dt \in {\cal A}^{(0,2)}$. Then $\delta F$ can be represented as 
\be
\delta F=\sum_kA_k\,\delta {u_k}\wedge dx\wedge dt+d\tilde F
\ee
where $\tilde F$ belongs to ${\cal A}^{(1,1)}$ (modulo $d$). The coefficient $A_k$ is uniquely determined for each species $k$. It will be denoted $\frac{\delta F}{\delta u_k}$ and called the variational derivative of $F$ with respect to $u_k$.
\end{proposition}
\begin{proof}
It is useful to sketch the proof of the first part of the claim as it clearly shows the connection with the usual variational principle and because we will repeatedly perform the procedure shown here in our examples below. One computes
\be
\delta F=\sum_{k,(\mu)}\frac{\partial f}{\partial u_k^{(\mu)}}\,\delta u_k^{(\mu)}\wedge dx \wedge dt 
\ee
and simply uses integration by parts repeatedly to write
\be
\sum_{k,(\mu)}\frac{\partial f}{\partial u_k^{(\mu)}}\,\delta u_k^{(\mu)}=\sum_{k,(\mu)}(-1)^{|\mu|}\partial^{(\mu)}\frac{\partial f}{\partial u_k^{(\mu)}}\,\delta u_k+\partial_x B_x +\partial_t B_t
\ee
for some $B_x$, $B_t$ in ${\cal A}^{(1,0)}$. Then, simply set 
\be
\label{def_Ak}
A_k=\sum_{(\mu)}(-1)^{|\mu|}\partial^{(\mu)}\frac{\partial f}{\partial u_k^{(\mu)}}\,,~~\tilde F=B_x\wedge dt-B_t\wedge dx\,.
\ee
The uniqueness of $A_k$ requires the use of the so-called Tulczyjev operator. We refer the reader to Proposition 19.4.4 in \cite{Dickey}.
\end{proof}

\subsection{Multisymplectic form and a covariant Poisson bracket}\label{multi}

In practice, we use the previous proposition in the case where $F$ is a Lagrangian volume form $\Lambda$ associated to a Lagrangian density $\cL$ describing the field theory at hand. We assume that the Lagrangian density $\cL$ depends on the fields $u_k$, $k=1, \dots, N$, and their derivatives up to the order $m$, \ie $\cL=\cL(u_k^{(\mu)})$ with $|\mu|=\mu_1+\mu_2\le m$. In that context, the form $\tilde F$ acquires an important role and we denote it by  $-\omegaone$ \ie we have 
\be
\label{Lambda_variation}
\delta \Lambda = \sum_k \varder{\cL}{u_k}\,\delta u_k \wedge dx \wedge dt  - d \omegaone\,.
\ee 
\begin{definition}
	The multisymplectic form $\Omega$ associated to the Lagrangian volume form $\Lambda$ is defined by 
	\be
\Omega = \delta \omegaone\,,	
	\ee
	where $\omegaone$ is the form obtained in \eqref{Lambda_variation}.
\end{definition}
A few remarks are in order. The form $\omegaone$ is determined up to a form of the type $d\omega^{1,0}$ so that the multisymplectic form is defined up to a form of the type $\delta d\omega^{1,0}$. It is known that adding a total derivative to the Lagrangian density gives rise to the same equation of motion. Here, if $d\omega^{0,1}$ is added to $\Lambda$ then $\omegaone$ acquires an additional term $\delta \omega^{0,1}$. However, the latter leaves the multisymplectic form $\Omega$ unchanged.
Of course, the equations of motion of theory are given by the Euler-Lagrange equations for $\Lambda$ which are 
\be
\varder{\Lambda}{u_k}=0\,,~~k=1,\dots,N\,.
\ee
Recall from \eqref{def_Ak} that $\omegaone$ can be written 
\be
\omegaone=\Omega^{(1)}_t\wedge dx-\Omega^{(1)}_x\wedge dt\,,
\ee
where $\Omega^{(1)}_{x,t}$ are (vertical) one forms in ${\cal A}^{(1,0)}$. Therefore 
\be
\label{def_multi_form}
\Omega=\delta\Omega^{(1)}_t\wedge dx-\delta\Omega^{(1)}_x\wedge dt\equiv \Omega_t\wedge dx-\Omega_x\wedge dt\,,
\ee
where $\Omega_{x,t}$ are (vertical) two forms whose explicit form depends on the field content of $\cL$ and of its highest jet dependence $m$.
Kanatchikov's idea \cite{K} is to mimic the well-known relation in classical mechanics. Given a non-degenerate, closed two forms and given a (Hamiltonian) function $F$ on the phase space, one can define a vector field $\xi_F$ by
\be
\label{relation_form_vector}
dF=\ip{\xi_F}{\omega}\,.
\ee
In particular, such a vector field always preserves the symplectic $\omega$ since, by Cartan's magic formula, the Lie derivative of $\omega$ along $\xi_F$ is given by
\be
L_{\xi_F}\omega=d(\ip{\xi_F}{\omega})+\ip{\xi_F}{d\omega}=0\,.
\ee
Conversely, the same formula shows that if the vector field $\xi$ is such that $L_{\xi}\omega=0$ then $d(\ip{\xi_F}{\omega})=0$ so by Poincar\'e's lemma, there exists (at least locally) a function $F$ on the phase space such that \eqref{relation_form_vector} holds for $\xi$.
One can define the Poisson bracket of two functions $F$ and $G$ on the phase space by setting
\be
\label{def_PB}
\{F,G\}=- \ip{\xi_F}{dG}= \omega(\xi_F,\xi_G)\,.
\ee
The fact that $\omega$ is closed has two related important consequences. Firstly,
\be
\xi_{\{F,G\}}=[\xi_F,\xi_G]\,,
\ee
where the bracket on the right hand side is the Lie bracket of two vector fields. Secondly, the Jacobi identity holds for the Poisson bracket $\{~,~\}$. 

In the multisymplectic setting, in order to generalise \eqref{def_PB}, Kanatchikov's proposal requires to first generalise \eqref{relation_form_vector} and to use the vertical derivation $\delta$ (denoted $d^V$ in \cite{K}) instead of $d$. The natural proposal is
\be
\label{relation_form_vector2}
\delta F=\ip{\xi_F}{\Omega}\,,
\ee
where $\Omega$ is the multisymplectic form of interest. Some important differences arise compared to the standard case. Firstly, not only is it possible to have functions in ${\cal A}^{(0,0)}$ on the left hand side of \eqref{relation_form_vector2} but it is also possible to have forms $F$ in ${\cal A}^{(0,1)}$ or ${\cal A}^{(0,2)}$ in principle. 
Given such a form $F$, the analog of  the problem of finding a Hamiltonian vector field, \ie the analog of \eqref{relation_form_vector}, becomes the problem of finding $\xi_F$ such that \eqref{relation_form_vector2} holds.

In turn, this requires the possibility to use more general vectors fields, or multivector fields, that can combine vertical and horizontal components. Thus, in addition to vector fields as in \eqref{vector_field}, in general we may use linear combinations (with coefficient in ${\cal A}$) of the following multivector fields
\be
\partial_{u_{i_1}^{(\mu_1)}} \^ \dots \^ \partial_{u_{i_p}^{(\mu_p)}} \^ \partial_{\nu_1} \^ \dots \^ \partial_{\nu_q}\,.
\ee
In our case, $q$ is at most $2$. In general, the existence of $\xi_F$ is not guaranteed in the multisymplectic setting and detailed investigation is required \cite{FPR}. This motivates the definition of Hamiltonian forms below.
Secondly, the multisymplectic form $\Omega$ is degenerate in general so that a (multi)vector field corresponding to a given form $F$ is not unique. However, if it exists, adding an element of the kernel of $\Omega$ to it will not change the result for the covariant Poisson bracket we define below. Therefore, in this paper, we always work modulo this kernel and talk about ``the'' vector field associated to a Hamiltonian form as a shorthand for a representative of the equivalence class of this vector field modulo the kernel of $\Omega$.
In view of this discussion, we need to define a class of forms $F\in {\cal A}^{(0,q)}$ with $q=0,1$ or $2$ for which a (multi)vector field can be found.
\begin{definition}\label{def_Ham_form}
A form $F$ is said to be Hamiltonian (with respect to $\Omega$) if there exists a (multi)vector field $X$ such that 
	\begin{equation}
	\ip{X}{\Omega} = \delta F\,.
	\end{equation}
	In that case, $X$ is called the Hamiltonian vector field related to $F$\footnote{The use of the definite article ``the'' is to be understood modulo the kernel of $\Omega$ of course, as discussed before.}. 
\end{definition}
In this paper, we will only need to consider forms in ${\cal A}^{(0,0)}$ (zero forms) or in ${\cal A}^{(0,1)}$ (one forms). Let us denote by $S_\Omega$ the set of basis elements $\delta u_k^{(\mu)}$ that appear explicitly the multisymplectic form. It is a finite set since $\Omega$ is derived from $\cL$ which is assumed to depend on $u_k^{(\mu)}$ with $|\mu|\le m$ for some $m$ (finite jet dependence). Hence, we can assume some ordering on $S_\Omega$ such that we can label the $\delta u_k^{(\mu)}$'s as $\delta v_j$, $j=1,\dots,|S_\Omega|$. We then write
\begin{equation}\label{msform:localcoordinates}
\Omega= \sum_{\substack{i<j\\i,j \in I}}\omega_x^{ij} \delta v_i \wedge \delta v_j \wedge dt - \sum_{\substack{i<j\\i,j \in J}}\omega_t^{ij}\delta v_i \wedge \delta v_j \wedge dx
\end{equation}
for some $I,J \subseteq \{1,\dots,|S_\Omega|\}$. Note that each $\omega_{x,t}^{ij}\in{\cal A}$ so has a dependence on the local coordinates $u_k^{(\mu)}$ which we do not show explicitly, and that in every example that we present they are non-degenerate (and therefore invertible).
\begin{proposition}\label{dependence_coords}  {\bf Necessary form of a Hamiltonian one-form.}\\
Suppose $F= F_1\, dx + F_2\, dt$, $F_{1,2}\in{\cal A}$ is a Hamiltonian form for the multisymplectic form \eqref{msform:localcoordinates}. Then,  $F_1$ can only depend (at most) on $v_j$, $j \in J$, and $F_2$ can only depend (at most) on $v_i$, $i \in I$.
\end{proposition}
\begin{proof}
Assume $F_1$ depends on some $u_\ell^{(\nu)}\notin \{v_j;j\in J\}$. On the one hand, 
\begin{equation}
\delta F = \sum_{j\in J} \parder{F_1}{v_j} \delta v_j \wedge dx + \parder{F_1}{u_\ell^{(\nu)}}\delta u_\ell^{(\nu)} \wedge dx + \sum_{i\in I} \parder{F_2}{v_i} \delta v_i \wedge dt\,.
\end{equation}
On the other, since $F$ is Hamiltonian, there exists a vector field $X$ such that $X \lrcorner \Omega = \delta F$. This gives
\begin{equation}
\sum_{\substack{i<j\\i,j \in I}}\omega_x^{ij} X \lrcorner \left(\delta v_i \wedge \delta v_j \wedge dt\right) - \sum_{\substack{i<j\\i,j \in J}}\omega_t^{ij}X \lrcorner \left(\delta v_i \wedge \delta v_j \wedge dx\right)
\end{equation}
In particular, this requires 
\begin{equation}
\sum_{j\in J} \parder{F_1}{v_j} \delta v_j \wedge dx + \parder{F_1}{u_\ell^{(\nu)}}\delta u_\ell^{(\nu)} \wedge dx =- \sum_{\substack{i<j\\i,j \in J}}\omega_t^{ij}X \lrcorner \left(\delta v_i \wedge \delta v_j \wedge dx\right)\,,
\end{equation}
so that necessarily $\displaystyle \parder{F_1}{v_j}=- \sum_{i \in J}\omega_t^{ij}X \lrcorner \delta v_i$ and $\parder{F_1}{u_\ell^{(\nu)}}=0$.
The same argument holds for $F_2$.
\end{proof}
Equipped with the notion of Hamiltonian forms, we can now define the covariant Poisson bracket of two such forms. 
\begin{definition}\label{def_cov_PB}
Let $F$ be a Hamiltonian $p$-form, $G$ be a Hamiltonian $q$-form, $p,q\in\{0,1\}$, and $X_F$ and $X_G$ be their Hamiltonian vector fields. The {\bf covariant Poisson bracket} of $F$ and $G$ is defined by 
\begin{equation}
\label{def_cov_PB}
\cpb{F}{G} = (-1)^{2-p} \ip{X_F}{\delta G}=(-1)^{2-p} \ip{X_F}{\ip{X_G}{\Omega}}\,.
\end{equation}
\end{definition}
One can show that that the covariant Poisson bracket satisfies graded anticommutativity  and graded Jacobi identity \cite{K}. Let $F$ be a Hamiltonian $p$-form, $G$ be a Hamiltonian $q$-form and $H$ be a Hamiltonian $r$-form. Then\footnote{Note that since $p,q\in\{0,1\}$ in our case, we can simplify the sign in the Jacobi identity. We can also check that the covariant Poisson bracket of two Hamiltonian forms is also a Hamiltonian form so that the Jacobi identity makes sense.},
\begin{equation}
\cpb{F}{G} = -(-1)^{g_1g_2}\cpb{G}{F}\,,
\end{equation}
\begin{equation}
\cpb{F}{\cpb{G}{H}} + \cpb{G}{\cpb{H}{F}}+\cpb{H}{\cpb{F}{G}}=0\,,
\end{equation}
with $g_1=1-p$, $g_2=1-q$.

\subsection{Elements of the classical $r$-matrix theory}\label{r_mat}

In the historic approach to the classical $r$-matrix, the starting point is to combine the Hamiltonian description of an integrable classical field theory, in particular its Poisson bracket, with its Lax pair formulation where the equations of motion are seen as a partial differential equations that one can rewrite as the zero curvature condition, or flatness condition, of the Lax connection $W=U\,dx+V\,dt$ describing the linear auxiliary problem 
\be
\begin{cases}
	\partial_x \Psi=U\,\Psi\,,\\
	\partial_t \Psi=V\,\Psi\,.
\end{cases}
\ee
Here, it should be understood that the so-called Lax pair $(U,V)$ is of Zakharov-Shabat type \cite{ZS} \ie $U$ and $V$ are matrices depending on the spacetime variables $x,t$ through the fields of the model at hand and also on the spectral parameter $\lda$ as a (Laurent) polynomial. All the examples we consider in the present article will be of this type. We will only need $2\times 2$ matrices (scalar field theories). It is a remarkable feature of Lax integrable partial differential equations that they are also (infinite dimensional) Hamiltonian system integrable in the Liouville sense, see \cite{ZaF,ZM2} for the first two historical examples. It is well known that the Lax pair for a given integrable field theory is not unique. Nevertheless, once a Lax pair is picked for the theory of interest, we will speak of {\it the} Lax connection (or one-form) of the theory. 

In the search for the canonical quantization of the inverse scattering method \cite{GGKM,ZS,AKNS}, Sklyanin made the following discovery \cite{Skly,Skly_r}. The (canonical) Poisson brackets of the fields of the integrable field theory can be equivalently rewritten using the space Lax matrix $U$ evaluated on the Cauchy surface in the following form
\be
\label{linear_r}
\{U_1(x,\lda),U_2(y,\mu)\}=\delta(x-y)\,[r_{12}(\lda,\mu),U_1(x,\lda)+U_2(y,\mu)]\,.
\ee
In our case, the Cauchy surface is simply the initial data surface at $t=0$ so that we display the space variable explicitly. This also motivates our calling $U$ the {\it space} Lax matrix as well as denoting the present Poisson bracket $\{~,~\}$ by $\{~,~\}_S$ and calling it space (or equal-time) Poisson bracket below. This will become further justified when we introduce the time (or equal-space) Poisson bracket $\{~,~\}_T$ and the covariant Poisson bracket $\cpb{~}{~}$ which combines $\{~,~\}_S$ and $\{~,~\}_T$ in an elegant way.

Some comments on the notation and the significance of \eqref{linear_r} are needed. The indices $1$ and $2$ are usually referred to as the auxiliary space notation. In our case, we will use the simplest instance whereby the notation $U_1$ means that we take the tensor product of the $2\times 2$ matrix $U$ in the first space with the $2\times 2$ identity matrix in the second space
\be
U_1=U\otimes \1\,.
\ee
Similarly, $U_2=\1\otimes U$. The object $r_{12}(\lda,\mu)$ is the central piece of this formalism and is called the classical $r$-matrix. The indices $12$ indicate that it lives in the tensor product of the space of $2\times 2$ matrices with itself and it has a functional dependence on the two spectral parameters $\lda,\mu$ (rational in our cases). Therefore, the right-hand side of \eqref{linear_r} is simply the commutator of $4\times 4$ matrices. The left-hand side should be understood as the $4\times 4$ matrix containing all possible Poisson brackets of the entries of $U(x,\lda)$ with the entries of $U(y,\mu)$. Thus, by definition, using $E_{ij}$ as the basis of $2\times 2$ matrices, we have\footnote{Summation over repeated indices is implied.}
\be
\{U_1(x,\lda),U_2(y,\mu)\}=\{U_{ij}(x,\lda),U_{kl}(y,\mu)\}\,E_{ij}\otimes E_{kl}\,.
\ee
For our purposes, all matrices involved will take values in the algebra $\text{sl}(2,\CC)$ so we will use instead its basis of Pauli matrices $\sigma_i$ with $i=1,2,3$ or $i=+,-,3$ depending on the model of interest. Hence, we will have
\be
\{U_1(x,\lda),U_2(y,\mu)\}=\{U_{i}(x,\lda),U_{j}(y,\mu)\}\,\sigma_i \otimes \sigma_j\,.
\ee
The significance of \eqref{linear_r} is that it represents the starting point of the abstract theory of Lie bialgebras and Poisson-Lie groups \cite{Drin} and of dressing actions \cite{STS} which form the unifying framework for the Hamiltonians properties of classical integrable systems. In the book \cite{FT} a detailed account of the use of the classical $r$-matrix method in conjunction with the inverse scattering method to obtain the Liouville integrability of certain integrable scalar field theories is given.

\section{The main result: A covariant Poisson bracket with $r$-matrix structure}\label{main_result}

In this section, for the reader's convenience, we present the main results of this paper in a synthetic form, with the important proviso that they have only been systematically obtained on all the examples detailed in the next section. In particular, an abstract formulation of a covariant theory of the classical $r$-matrix that would combine elements of the work of Semenov-Tian-Shansky \cite{STS,STS2} and the geometric formalism of the calculus of variations is not available yet and is left for future investigation.

With this in mind, let us start with a Lagrangian (volume) form for an integrable field theory with a Zakharov-Shabat Lax pair $(U(\lda),V(\lda))$
\be
\Lambda=\Lag\,\vol\,,
\ee
and derive from it the multisymplectic form $\Omega$ as explained in Section \ref{multi}. With $\Omega$, we define our covariant Poisson bracket $\cpb{~}{~}$ as in Definition \ref{def_cov_PB}. For convenience, we will call the Lax connection $W(\lda)=U(\lda)\,dx+V(\lda)\,dt$ the Lax form of the field theory as we will systematically view it as a ($\lda$-dependent) one-form in ${\cal A}^{(0,1)}$. 

Before we formulate the main theorem below, we present some results that will be needed to obtain an alternative proof of the theorem that is more elegant than the direct explicit calculation we performed for each of our examples. 
Recall that the multisymplectic form $\Omega$ \eqref{def_multi_form} derived from a Lagrangian volume form can be written as
\be
\label{split_Omega}
\Omega= \Omega_x\wedge dt-\Omega_t\wedge dx\,.
\ee
It turns out that $\Omega_t$ and $\Omega_x$ are bona fide symplectic forms: nondegenerate, (vertically) closed forms. Therefore, each of them can be used individually to define a Poisson bracket in the standard way, that we will call a ``single time'' Poisson bracket, for reason that will become clear. In local coordinates,using the notations from \eqref{msform:localcoordinates} we can write
\be
\Omega_t=\sum_{\substack{i < j\\ i,j \in J}} \omega_t^{ij}\,\delta v_i\wedge\delta v_j\,,~~\Omega_x=\sum_{\substack{i < j\\ i,j \in I}}\omega_x^{ij}\,\delta v_i\wedge\delta v_j\,.
\ee
\begin{definition}\label{def_single_time}
The single time Poisson brackets $\{~,~\}_S$ and $\{~,~\}_T$ are defined by the (vertical) Poisson bivectors 
\be
P^S=\sum_{\substack{i < j\\ i,j \in J}}\pi^S_{ij}\,\partial_{v_i}\wedge\partial_{v_j}\,,~~P^T=\sum_{\substack{i < j\\ i,j \in I}}\pi^T_{ij}\,\partial_{v_i}\wedge\partial_{v_j}\,,
\ee
where $\pi^S$ (resp. $\pi^T$) is the inverse of the matrix $\omega_t$ (resp. $\omega_x$).
\end{definition}
The reason for the notation $\{~,~\}_S$ and $\{~,~\}_T$ comes from the fact that the Poisson brackets so defined provide precisely a finite-dimensional version of the two Poisson brackets on infinite dimensional phase space derived in \cite{CK,C,ACDK} from the standard Legendre transformation (with respect to the time-derivative of the fields) and its accompanying covariant companion (with respect to the space-derivative of the fields). This will be made explicit in Section \ref{Examples} containing the examples. A striking feature is that in the works \cite{CK,ACDK}, a Dirac procedure was required to obtain $\{~,~\}_S$ and $\{~,~\}_T$ as the Lagrangian in the AKNS hierarchy are all degenerate. But here, the procedure explained above to derive the multisymplectic form $\Omega$ delivers $\Omega_t$ and $\Omega_x$ as well as $\{~,~\}_S$ and $\{~,~\}_T$ directly, with no need for a Dirac procedure. We do not have an explanation for this remarkable observation yet but we only mention that it might provide in the present multisymplectic context the analog of the argument popularised by Faddeev and Jackiw in \cite{FJ}. This deserves further investigation that is beyond the scope of the present paper.

Equipped with this, we have the following proposition that shows that the splitting \eqref{split_Omega} of the multisymplectic form has a counterpart at the level of the covariant Poisson bracket.
\begin{proposition}\label{splitting}
Let $F=A\,dx + B\,dt$ and $G=C\,dx + D\,dt$ be two Hamiltonian 1-forms. Then,
\begin{equation}
\label{eq_split}
\cpb{F}{G} =  \{B,D\}_T\, dt-\{A,C\}_S\,dx  
\end{equation} 
where the two single-time Poisson Brackets $\{~,~\}_S$ and $\{~,~\}_T$ are as in Definition \ref{def_single_time}.
\end{proposition}
Note that the components $A,B,C,D$ of the Hamiltonian forms depend of course on the local coordinates according to Proposition \ref{dependence_coords} but may also depend on the spectral parameter $\lda$, as is the case for instance if they equal matrix entries of $U(\lda)$ or $V(\lda)$, the components of the Lax form. The spectral parameter is always treated a non dynamical variable which Poisson commutes with everything else. In order to formulate the main result, in addition to a spectral parameter dependence, we also need to extend the auxiliary space notation of Section \ref{r_mat} to the covariant Poisson bracket. Confining ourselves to $\text{sl}(2,{\cal A})$ matrices for convenience, we can do this componentwise by choosing a basis.
\begin{definition}\label{extension_lda}
	Given two $\lda$-dependent 1-forms with $\text{sl}(2,{\cal A})$-valued coefficients, $F(\lda)=A(\lda)\,dx + B(\lda)\,dt$ and $G(\lda)=C(\lda)\,dx + D(\lda)\,dt$, where $A(\lda)=A^i(\lda)\,\sigma_i$ and similarly for $B(\lda),C(\lda),D(\lda)$, we say that they are Hamiltonian if $F^i(\lda)=A^i(\lda)\,dx+B^i(\lda)\,dt$ and $G^i(\lda)=C^i(\lda)\,dx + D^i(\lda)\,dt$, $i=1,2,3$ are Hamiltonian one-forms in the sense of Definition \ref{def_Ham_form}, extended pointwise in $\lda$. In this case, we define 
	\be	
	\label{def_matrix_cov_PB}
	\cpb{F_1(\lda)}{G_2(\mu)}\equiv \cpb{F^i(\lda)}{G^j(\mu)}\,\sigma_i\otimes \sigma_j\,.
	\ee
\end{definition}
Equipped with this definition, we can now formulate the main result of this paper: the covariant Poisson bracket structure of the integrable field theory under consideration is governed by the same classical $r$-matrix that governs the space Poisson bracket in the traditional (non covariant) Hamiltonian approach to the integrable field theory.
\begin{theorem}
	\label{main_th}
The Lax form $W(\lda)$ is a Hamiltonian one-form with respect to $\Omega$. It satisfies
\be
\label{cov_r_PB}
\cpb{W_1(\lda)}{W_2(\mu)}=[r_{12}(\lda,\mu),W_1(\lda)+W_2(\mu)]\,,
\ee
where the right-hand side is understood as 
\be
\label{def_RHS}
[r_{12}(\lda,\mu),W_1(\lda)+W_2(\mu)]=[r_{12}(\lda,\mu),U_1(\lda)+U_2(\mu)]\,dx +[r_{12}(\lda,\mu),V_1(\lda)+V_2(\mu)]\,dt\,.
\ee
\end{theorem}
As illustrated on the examples in the next Section, this result is obtained by direct calculation which is facilitated by the splitting property of Proposition \ref{splitting}. In fact, thanks to Definition \ref{extension_lda}, we can readily extend the validity of Proposition \ref{splitting} to $\lda$-dependent 1-forms with $\text{sl}(2,{\cal A})$-valued coefficients. 
In particular, we have the following corollary
\begin{corollary}\label{UU_VV}
The components $U(\lda)$ and $V(\lda)$ of the Lax form satisfy
\bea
&&\{U_1(\lambda),U_2(\mu)\}_S= -[r_{12}(\lda,\mu),U_1(\lda)+U_2(\mu)]\,,\\
&&\{V_1(\lambda),V_2(\mu)\}_T= [r_{12}(\lda,\mu),V_1(\lda)+V_2(\mu)]\,.  
\eea
where $\{~,~\}_{S,T}$ are the two single time Poisson brackets of Definition \ref{def_single_time}.
\end{corollary}
\begin{proof}
	A direct consequence of the splitting formula \eqref{eq_split} is 
	\be
	\cpb{W_1(\lda)}{W_2(\mu)}=\{V_1(\lambda),V_2(\mu)\}_T\, dt-  \{U_1(\lambda),U_2(\mu)\}_S\,dx\,.
	\ee
	It remains to compare with \eqref{cov_r_PB}-\eqref{def_RHS}.
\end{proof}
It is well known that the validity of the Jacobi identity for a Poisson bracket given by an $r$-matrix structure is ensured by the fact that $r$ satisfies the so-called classical Yang-Baxter equation
\begin{eqnarray}
\label{cYBE}
[\ r_{13}(\lda,\nu)\ ,\ r_{23}( \mu,\nu)\ ]+ [\ r_{12}(\lda,\mu)\ ,\ r_{13}(\lda,\nu)\ ]+ [\ r_{12}(\lda,\mu)\ ,\ r_{23}( \mu,\nu)\ ]=0\,.
\end{eqnarray}
The same holds for our covariant $r$-matrix structure in \eqref{cov_r_PB}. In our examples, two of the most famous solutions of \eqref{cYBE} will be used: the trigonometric one (for sine-Gordon) and the rational one (for NLS and modified KdV). 

The single-time Poisson structures of Corollary \ref{UU_VV} constitute the main observation initially made in \cite{CK,C}, further investigated in \cite{ACDK} and established in more generality in \cite{AC}. Unifying them into a coherent covariant Poisson bracket was the main motivation for the present work. Our construction reproduces them as a by-product as desired. A few comments are in order though. Here, we have obtained them as a byproduct of the covariant Poisson bracket formalism which treat a field theory as a finite dimensional system. In other words, the fibres with local coordinates $u_k$, $k=1,\dots,N$ are finite dimensional manifolds. In contrast, in \cite{CK,C,ACDK,AC}, the infinite dimensional point of view of the Hamiltonian formalism was taken and the Legendre transform was taken separately with respect to the time and space variable, resulting in
\bea
\label{UU}
&&\{U_1(x,\lda),U_2(y,\mu)\}_S=\delta(x-y)\,[r_{12}(\lda,\mu),U_1(x,\lda)+U_2(y,\mu)]\,,\\
\label{VV}
&&\{V_1(t,\lda),V_2(\tau,\mu)\}_T=-\delta(t-\tau)\,[r_{12}(\lda,\mu),V_1(t,\lda)+V_2(\tau,\mu)]\,,
\eea
with the $\delta$ distribution characteristic of infinite dimensional phase spaces. Another comment is that the reader might wonder why there is an overall minus sign between the Poisson structures in Corollary \ref{UU_VV} and \eqref{UU}-\eqref{VV}. This is due to a different convention between the present paper, where we followed Kanatchikov's sign convention in \eqref{def_cov_PB}, and the conventions used \cite{CK,C,ACDK,AC}. This amounts to changing $r$ to $-r$ which is of no consequence. However, the relative sign between between \eqref{UU} and \eqref{VV} is derived consistently in the present approach. It was unexplained originally in \cite{CK,C} but now finds an explanation in the form of the splitting property: it is dictated by the multisymplectic formalism behind the definition of the covariant Poisson bracket and its connection to single time structures. 

Let us now review the covariant Hamiltonian description of a field theory, thereby justifying the terminology covariant Poisson bracket and covariant Hamiltonian formalism used so far. The following is based on a combination of ideas and objects that can be found in \cite{K} and \cite{Dickey}. Let us introduce the energy-momentum tensors 
\begin{equation}
\label{EM_tensor}
T_\nu= - \ip{\partial_\nu}{ \Lambda }+ \ip{\tilpartial_\nu}{ \Omega^{(1)}}, \qquad  \nu =x,t,
\end{equation}
where the vector field $\tilpartial_\nu$ is defined by
\be
\tilpartial_\nu=\sum_{k,(\mu)}u_k^{(\mu)+e_\nu}\,\partial_{u_k^{(\mu)}}\,.
\ee
These tensors are (horizontal) one-forms that we can write as
\be
T_x= T_{xx}\, dt - T_{xt}\, dx\,,~~T_t= T_{tx}\, dt - T_{tt}\, dx\,.
\ee
The covariant Hamiltonian is the (horizontal) two-form defined by
\be
\mathcal{H}=dx\wedge \ip{\tilpartial_x}{\omegaone}+dt\wedge \ip{\tilpartial_t}{\omegaone}-\Lambda\,.
\ee
Then one can show that if the Lagrangian form does not depend on the spacetime variables explicitly (autonomous systems) then the equation of motion $\varder{\Lambda}{u_k}=0$, $k=1,\dots,N$ are equivalent to the following covariant Hamilton's equations
\be
\delta \mathcal{H}=dx\wedge \ip{\tilpartial_x}{\Omega}+dt\wedge \ip{\tilpartial_t}{\Omega}\,.
\ee
However, the covariant nature of ${\cal H}$ and $\cpb{~}{~}$ is better appreciated with the following form of the equation of motion. 
Let us write 
\be
{\cal H}=H\,dx\wedge dt\,.
\ee
We can check that $H$ is always a Hamiltonian (zero) form so that we can always find a Hamiltonian vector field $X_H$ for it. Then, for each example in the next section, we can show the following.
\begin{proposition}
If $F$ is any Hamiltonian one-form then the Euler-Lagrange equations $\varder{\Lambda}{u_k}=0$, $k=1,\dots,N$ imply
\be
dF = \cpb{H}{F}\,dx\wedge dt\,.
\ee	
\end{proposition}
To obtain the covariant Hamiltonian description of the Euler-Lagrange equations, we need a converse to this statement. In all the examples considered in this paper, we can explicitely establish the following fact.
\begin{proposition}
The Euler-Lagrange equations of motion are equivalent to the covariant Hamilton equation of motion for the Lax form\footnote{Here, in line with Definition \ref{extension_lda}, $\cpb{H}{W(\lda)}$ means $\cpb{H}{W^i(\lda)}\sigma_i$.}
\be
dW(\lda)=\cpb{H}{W(\lda)}\,dx\wedge dt\,.
\ee
Moreover, we have\footnote{The notation $W(\lda)\wedge W(\lda)$ is a shorthand for the usual operation on Lie algebra-valued forms. In our case, we have an associative product for the components of $W(\lda)$ (matrix multiplication). In short, $W(\lda)\wedge W(\lda)=\left(U(\lda)V(\lda)-V(\lda)U(\lda)\right)dx\wedge dt=\left[U(\lda),V(\lda) \right]dx\wedge dt$}
\be
\cpb{H}{W(\lda)}=W(\lda)\wedge W(\lda)\,.
\ee
Therefore the Maurer-Cartan equation characterising the zero curvature condition of the Lax connection, \ie $U_t(\lda)-V_x(\lda)+[U(\lda),V(\lda)]=0$, is derived as a covariant Hamiltonian equation for the Hamiltonian $H$ and the covariant Poisson bracket $\cpb{~}{~}$.
\end{proposition}
The relation between the Maurer-Cartan equation and a Lax connection is well-known of course. However, this is the first time that this is derived as a covariant Hamilton equation.

\section{Examples}\label{Examples}

\subsection{A relativistic example: the Sine-Gordon model in laboratory coordinates}

The sine-Gordon model for the real scalar field $\phi(x,t)$ reads
\begin{equation}\label{SG}
\phi_{tt}-\phi_{xx} + \frac{m^2}{\beta} \sin \beta \phi=0\,.
\end{equation}
A Lagrangian form for it is given by
\begin{equation}
\Lambda= [ \frac{1}{2}(\phi_t^2-\phi_x^2)- \frac{m^2}{\beta^2}(1-\cos \beta \phi)]\, \vol\,.
\end{equation}
Equation \eqref{SG} is equivalent to the following zero curvature equation which should hold as an identity in $\lda$
\be
\partial_t U(\lda)-\partial_x V(\lda)+[U(\lda),V(\lda)]=0\,.
\ee
where the Lax pair $(U,V)$ can be taken as
\begin{equation}
U(\lambda)= -ik_0(\lambda) \sin{\frac{\beta \phi}{2}} \sigma_1 - i k_1(\lambda) \cos{\frac{\beta \phi}{2}}\sigma_2 - \frac{i\beta}{4}\phi_t \sigma_3\,,
\end{equation}
\begin{equation}
V(\lambda)= -ik_1(\lambda) \sin{\frac{\beta \phi}{2}} \sigma_1 - i k_0(\lambda) \cos{\frac{\beta \phi}{2}}\sigma_2 - \frac{i\beta}{4}\phi_x \sigma_3\,,
\end{equation}
where $k_0(\lambda) = \frac{m}{4}(\lambda + \lambda^{-1})$ and $ k_1(\lambda) = \frac{m}{4}(\lambda - \lambda^{-1})$. 
In the general notations of Section \ref{generalities}, here $N=1$, $m=1$, and the only field is $u_1=\phi$. We will denote $u_k^{(i)}$, $(i)=(0,0)$, $(1,0)$, $(0,1)$, etc. as $\phi$, $\phi_x$, $\phi_t$, etc. for convenience. It is important to remember that $\phi_x$, $\phi_t$, etc. should be treated as coordinates in the differential algebra ${\cal A}$ when performing the calculations in the variational bicomplex.
\begin{proposition}
The form $\omegaone$ is given by 
\begin{equation}
\omegaone= -\phi_t \delta \phi \^ dx - \phi_x \delta \phi \^ dt\,.
\end{equation}
and the multisymplectic form reads
\begin{equation}
\Omega = - \delta \phi_t \wedge \delta \phi \wedge dx - \delta \phi_x \wedge \delta \phi \wedge dt\,.
\end{equation}
\end{proposition}
\begin{proof}
The $\delta$ variation of $\Lambda$ is
\begin{equation}
\delta \Lambda = [\phi_t \delta \phi_t - \phi_x \delta \phi_x -\frac{m^2}{\beta} \sin \beta \phi \delta \phi] \wedge\vol.
\end{equation}
Now, using \eqref{ddeltagenerator}, which in this case means $d(\delta \phi) = - \delta \phi_x \wedge dx -\delta \phi_t \wedge dt$, we get that $d(\phi_t \delta \phi \wedge dx) = \phi_{tt} dt \wedge \delta \phi \wedge dx + \phi_t d(\delta \phi) \wedge dx = \phi_{tt} \delta \phi \wedge \vol +\phi_t \delta \phi_t \wedge \vol$, and therefore 
\begin{equation}
\phi_t \delta \phi_t \wedge \vol = -\phi_{tt} \delta \phi \wedge \vol + d(\phi_t \delta \phi \wedge dx),
\end{equation}
and equivalently
\begin{equation}
-\phi_x \delta \phi_x \wedge \vol = \phi_{xx} \delta \phi \wedge \vol + d ( \phi_x \delta \phi \wedge dt).
\end{equation}
Therefore, the variation of $\Lambda$ brings
\begin{equation}
\delta \Lambda =[- \phi_{tt} + \phi_{xx} - \frac{m^2}{\beta} \sin \beta \phi ]\delta \phi \wedge \vol + d( \phi_t \delta \phi \^ dx + \phi_x \delta \phi \^ dt).
\end{equation}
By looking at $\varder{\Lambda}{\phi} =0$ we obtain the Sine-Gordon equation. $\omegaone$ then reads
\begin{equation}
\omegaone= -\phi_t \delta \phi \^ dx - \phi_x \delta \phi \^ dt.
\end{equation}
Its $\delta$-differential $\delta \omegaone$ is defined to be the multisymplectic form $\Omega$
\begin{equation}\label{msform}
\Omega= \delta \omegaone = - \delta \phi_t \wedge \delta \phi \wedge dx - \delta \phi_x \wedge \delta \phi \wedge dt.
\end{equation}
\end{proof}
Equipped with the multisymplectic form $\Omega$ we can define the covariant Poisson bracket and also the two ``single-time'' Poisson brackets.
\begin{proposition}
	A Hamiltonian one-form for the SG equation is $F=F^1 (\phi,\phi_t)dx + F^2(\phi,\phi_x)dt$ where
	\begin{equation}
	\parder{F^1}{\phi_t} = \parder{F^2}{\phi_x}
	\end{equation}
	The respective vector field is
	\begin{equation}
	X_F = \parder{F^1}{\phi_t} \partial_\phi - \parder{F^2}{\phi} \partial_{\phi_x} - \parder{F^1}{\phi} \partial_{\phi_t}\,.
	\end{equation}
	For any two Hamiltonian one-forms $F=Adx + Bdt$ and $G=Cdx + Ddt$, we have following splitting formula
	\begin{equation}
	\label{SG_split}
	\cpb{F}{G} = \{B,D\}_T \,dt- \{A,C\}_S\,dx  
	\end{equation} 
	where the single-time Poisson Brackets are given by
	\begin{equation}
	\{A, C \}_S = \parder{A}{\phi_t}\parder{C}{\phi}-\parder{A}{\phi} \parder{C}{\phi_t} , \qquad \{B, D \}_T = \parder{B}{\phi}\parder{D}{\phi_x} - \parder{D}{\phi}\parder{B}{\phi_x}.
	\end{equation}
\end{proposition}
\begin{proof}
	Let us consider the following (vertical) vector field
	\begin{equation}
	X_F = A\partial_\phi + B \partial_{\phi_x}+ C \partial_{\phi_t}
	\end{equation}
	in the equation $\delta F = \ip{X_F}{\Omega}$. The left hand-side reads 
	\begin{equation}
	\delta F = \parder{F^1}{\phi}\delta \phi \wedge dx + \parder{F^1}{\phi_t} \delta \phi_t \wedge dx + \parder{F^2}{\phi} \delta \phi \wedge dt + \parder{F^2}{\phi_x} \delta \phi_x \wedge dt,
	\end{equation}
	while the right hand-side is
	\begin{equation}
	\ip{X_F}{\Omega} = A \delta \phi_t \wedge dx + A \delta \phi_x \wedge dt - B \delta \phi \wedge dt - C \delta \phi \wedge dx.
	\end{equation}
	A direct comparison shows
	\begin{equation}
	A =  \parder{F^1}{\phi_t} = \parder{F^2}{\phi_x}, \qquad  B = -\parder{F^2}{\phi}, \qquad C= -\parder{F^1}{\phi}.
	\end{equation}
	Then, \eqref{SG_split} follows by a direct calculation from $\cpb{F}{G}=- i(X_F) \delta G$ and recognizing the single-time Poisson brackets as defined in the Proposition.
	\end{proof}

\begin{theorem}\label{covariantrmatrix:SG}
The Lax form $W(\lda)=U(\lda)\,dx+V(\lda)\,dt$ satisfies the following covariant Poisson bracket 
\be
\cpb{W_1(\lambda)}{W_2(\mu)}=[r_{12}(\lda,\mu),W_1(\lambda)+W_2(\mu)]
\ee
where the classical $r$-matrix is that of the sine-Gordon model (see e.g. \cite{FT})
\begin{equation}
r_{12}(\lambda, \mu) = f (\lambda, \mu) ( \id \otimes \id - \sigma_3 \otimes \sigma_3) + g(\lambda,\mu) (\sigma_1 \otimes \sigma_1 + \sigma_2 \otimes \sigma_2),
\end{equation} 
with $f(\lambda,\mu)= - \frac{\beta^2}{16} \frac{\lambda^2 + \mu^2}{\lambda^2 - \mu^2}$ and $g(\lambda, \mu) = \frac{\beta^2}{8} \frac{\lambda \mu}{\lambda^2 - \mu^2}$.
\end{theorem}
\begin{proof}
The proof is done by straightforward but long calculations. We give the details for this first example.
We write $W (\lambda)= \sum_i W^i(\lambda) \sigma_i$, where $W^i(\lambda) = U^i(\lambda)\, dx + V^i(\lambda)\, dt$, so that
\bea
&&W^1(\lambda) = -ik_0(\lambda) \sin \frac{\beta\phi}{2} dx -ik_1(\lambda) \sin \frac{\beta\phi}{2} dt\,,\\
&&W^2(\lambda) = -ik_1(\lambda) \cos \frac{\beta\phi}{2} dx -ik_0(\lambda) \cos \frac{\beta\phi}{2} dt\,,\\
&&W^3(\lambda) = - \frac{i\beta}{4} \phi_t dx- \frac{i\beta}{4} \phi_x dt\,.
\eea
It can be checked that $W^i$, $i=1,2,3$ are Hamiltonian forms. Therefore, using the splitting property \ref{SG_split}, we find that the only non-zero Poisson brackets are
\begin{gather}
\cpb{W^1(\lambda)}{W^3(\mu)} = - \frac{\beta^2}{8} \cos \frac{\beta \phi}{2} ( k_0(\lambda) dx + k_1(\lambda) dt),\\
\cpb{W^2(\lambda)}{W^3(\mu)} =  \frac{\beta^2}{8} \sin \frac{\beta \phi}{2} ( k_1(\lambda) dx + k_0(\lambda) dt),\\
\cpb{W^3(\lambda)}{W^1(\mu)} =  \frac{\beta^2}{8} \cos \frac{\beta \phi}{2} ( k_0(\mu) dx + k_1(\mu) dt),\\
\cpb{W^3(\lambda)}{W^2(\mu)} = - \frac{\beta^2}{8} \sin \frac{\beta \phi}{2} ( k_1(\mu) dx + k_0(\mu) dt).
\end{gather}
Thus we deduce, according to the definition \eqref{def_matrix_cov_PB},
\begin{multline}
\label{expression_WW}
\cpb{W_1(\lambda)}{W_2(\mu)} =  \frac{\beta^2}{8} \left[-\cos \frac{\beta \phi}{2} ( k_0 (\lambda) dx + k_1(\lambda) dt)\, \sigma_1 \otimes \sigma_3 + \sin \frac{\beta \phi}{2} ( k_1 (\lambda) dx + k_0(\lambda) dt)\, \sigma_2 \otimes \sigma_3\right. \\ + \left. \cos \frac{\beta \phi}{2} ( k_0 (\mu) dx + k_1(\mu) dt)\, \sigma_3 \otimes \sigma_1 - \sin \frac{\beta \phi}{2} ( k_1 (\mu) dx + k_0(\mu) dt) \,\sigma_3 \otimes \sigma_2    \right].
\end{multline}
On the other hand, we can also compute $[r_{12}(\lambda - \mu) , W_1(\lambda) + W_2(\mu)]$ directly, using the commutation rules $[\sigma_i,\sigma_j] = 2i \epsilon_{ijk} \sigma_k$ and the property $[A \otimes \id, B\otimes C] = [A , B] \otimes C$. We find 
\begin{align*}
&[r_{12}(\lambda - \mu) , W_1(\lambda) + W_2(\mu)]\\
=&[-f(\lambda,\mu) \sigma_3 \otimes \sigma_3 + g(\lambda,\mu) \sigma_1 \otimes \sigma_1 + g(\lambda,\mu)\sigma_2 \otimes \sigma_2, W^1(\lambda) \sigma_1 \otimes \id\\&+ W^2(\lambda) \sigma_2 \otimes \id + W^3(\lambda) \sigma_3 \otimes \id + W^1(\mu) \id \otimes \sigma_1 + W^2(\mu) \id \otimes \sigma_2 + W^3(\mu) \id \otimes \sigma_3]\\
=& -2i (f(\lambda, \mu) W^1(\lambda) + g(\lambda, \mu) W^1(\mu))\sigma_2 \otimes \sigma_3 + 2i (f(\lambda, \mu) W^2(\lambda) +g(\lambda, \mu) W^2(\mu))\sigma_1 \otimes \sigma_3 \\&- 2i (f(\lambda, \mu) W^1(\mu) + g(\lambda, \mu) W^1(\lambda))\sigma_3 \otimes \sigma_2 - 2i (f(\lambda, \mu) W^2(\mu) + g(\lambda, \mu) W^2(\lambda))\sigma_3 \otimes \sigma_1\\&+ 2i ( g(\lambda, \mu) W^3(\mu) -g(\lambda, \mu) W^3(\lambda) )\sigma_2 \otimes \sigma_1 + 2i (g(\lambda, \mu) W^3(\lambda) - g(\lambda, \mu) W^2(\mu))\sigma_1 \otimes \sigma_2.
\end{align*}
Upon inserting the explicit expressions of $W^i$, $f$ and $g$ one recovers \eqref{expression_WW} and the claim is proved.
\end{proof}

We conclude this section on the sine-Gordon model with its covariant Hamiltonian formulation. 
We first compute the energy-momentum tensors $T_x= T_{xx} dt - T_{xt} dx$ and $T_t= T_{tx} dt - T_{tt} dx$ according to \eqref{EM_tensor} to find $T_{xx} = -\frac{1}{2} \phi_t^2 -\frac{1}{2}\phi_x^2 + \frac{m^2}{\beta^2}(1-\cos \beta \phi)$ and $ T_{tt}= \frac{1}{2} \phi_t^2 +\frac{1}{2}\phi_x^2 + \frac{m^2}{\beta^2}(1-\cos \beta \phi)$. The covariant Hamiltonian $\mathcal{H} = H \vol$ can be computed as $\mathcal{H}= (T_{xx} +T_{tt} + \Lambda) \vol$ and is given by
\begin{equation}
H= \frac{1}{2}(\phi_t^2 - \phi_x^2) + \frac{m^2}{\beta^2} (1- \cos \beta \phi).
\end{equation}
The corresponding Hamiltonian vector field $X_H$ can be taken as
\begin{equation}
X_H = \phi_t \partial_\phi \wedge \partial_x - \phi_x \partial_\phi \wedge \partial_t - \frac{m^2}{2\beta} \sin\beta \phi (\partial_{\phi_t} \wedge \partial_x + \partial_{\phi_x} \wedge \partial_t).
\end{equation}
Let us now consider the Lax Form $W(\lambda)=U(\lambda)dx+ V(\lambda)dt$. On the one hand, we have
\begin{multline}
d W(\lambda) = ( (- i k_1(\lambda) \cos \frac{\beta \phi}{2} \phi_x+ ik_0(\lambda) \cos \frac{\beta \phi}{2} \phi_t) \sigma_1 \\+(-i k_1(\lambda) \sin \frac{\beta \phi}{2} \phi_t+ i k_0(\lambda) \sin \frac{\beta \phi}{2} \phi_x) \sigma_2 +( \frac{i\beta}{4} \phi_{tt} - \frac{i\beta}{4} \phi_{xx} ) \sigma_3  )dx \wedge dt
\end{multline}
and on the other hand,
\begin{equation}
\begin{split}
\cpb{H}{W(\lambda)} =& \ip{X_H}{\delta W(\lambda)}\\
=& \ip{X_H} (( -i\frac{\beta k_0(\lambda)}{2} \cos \frac{\beta \phi}{2} \delta \phi \wedge dx - i \frac{k_1(\lambda) \beta}{2} \cos \frac{\beta \phi}{2} \delta \phi \wedge dt) \sigma_1\\ &+(i \frac{k_1(\lambda) \beta}{2} \sin \frac{\beta \phi}{2} \delta \phi \wedge dx + i \frac{k_0(\lambda) \beta}{2} \sin \frac{\beta \phi}{2} \delta \phi\wedge dt )\sigma_2 - i \frac{\beta}{4} (\delta \phi_t \wedge dx + \delta \phi_x \wedge dt)\sigma_3 )\\
=& ( \frac{i \beta}{2} (k_0 (\lambda) \phi_t - k_1 (\lambda) \phi_x) \cos \frac{\beta \phi}{2}) \sigma_1 \\
&+ ( \frac{i\beta}{2}(k_0 (\lambda) \phi_x - k_1(\lambda) \phi_t) \sin\frac{\beta \phi}{2}) \sigma_2 - i \frac{m^2}{4} \sin \beta \phi \sigma_3.
\end{split}
\end{equation}
Therefore
\be
d W(\lambda)=\cpb{H}{W(\lambda)}\,\vol \Leftrightarrow \phi_{tt}-\phi_{xx} + \frac{m^2}{\beta} \sin \beta \phi=0\,,
\ee
which is the desired covariant Hamiltonian form of the sine-Gordon equation. One can verify with a direct computation that $\cpb{H}{W(\lambda)} = [U(\lambda), V(\lambda)]$.

\subsection{Nonrelativistic examples: the nonlinear Schr\"odinger and modified KdV equations}

\subsubsection{Nonlinear Schr\"odinger equation}

By a slight abuse of language, we call the following system of equations for two complex scalar fields $q,r$ the nonlinear Schr\"odinger (NLS) equation
\be
\label{NLS}
\begin{cases}
iq_t + q_{xx} - 2 q^2 r =0\\
i r_t - r_{xx} + 2 r^2 q =0\,. 
\end{cases}
\ee
Strictly speaking, NLS appears under the reduction $r=\pm q^*$. A Lagrangian form for \eqref{NLS} is given by
\begin{equation}
\Lambda = [ \frac{i}{2}(r q_t - r_t q) - r_x q_x -   r^2 q^2]\, dx \wedge dt,
\end{equation}
The system \eqref{NLS} is equivalent to the zero curvature equation which must hold as an identity in $\lda$
\be
\partial_t U(\lda)-\partial_x V(\lda)+[U(\lda),V(\lda)]=0\,.
\ee
where the Lax pair $(U,V)$ can be taken as
\bea
&&U(\lambda)= -\frac{i \lambda}{2} \sigma_3  + q \sigma_+ + r \sigma_-\,,\\
&&V(\lambda)= \left( \frac{\lambda^2}{2i} - i qr\right) \sigma_3 + ( \lambda q + i q_x) \sigma_+ + (\lambda r - i r_x) \sigma_-\,.
\eea
In the general notations of Section \ref{generalities}, here $N=2$, $m=1$, and the two field are $u_1=q$ and $u_2=r$. We will denote $u_k^{(i)}$, $k=1,2$, $(i)=(0,0)$, $(1,0)$, etc. as $q$, $r$, $q_x$, $r_x$, etc. for convenience.
\begin{proposition}
The form $\omegaone$ is given by
\be
\omegaone=\frac{i}{2}(q \delta r - r \delta q)\wedge dx - (q_x \delta r + r_x \delta q) \wedge dt\,,
\ee
and the multisymplectic form reads
\begin{equation}
\Omega = i \delta q \wedge \delta r \wedge dx +(\delta r \wedge \delta q_x + \delta q \wedge \delta r_x) \wedge dt\,.
\end{equation}
\end{proposition}
\begin{proof}
The $\delta$-variation of the Lagrangian is
\begin{equation}
\delta \Lambda = [ \frac{i}{2}(\delta r q_t + r \delta q_t -\delta r_t q - r_t \delta q) - \delta r_x q_x - r_x \delta q_x  - 2   r \delta r q^2 - 2   r^2 q \delta q    ]\wedge dx \wedge dt.
\end{equation}
Then, using
\begin{gather}
\frac{i}{2} r \delta q_t \wedge dx \wedge dt = d( \frac{i}{2} r \delta q \wedge dx) - \frac{i}{2} r_t \delta q \wedge dx \wedge dt,\\
\frac{i}{2} q \delta r_t \wedge dx \wedge dt = d( \frac{i}{2} q \delta r \wedge dx) - \frac{i}{2} q_t \delta r \wedge dx \wedge dt,\\
- q_x \delta r_x \wedge dx \wedge dt = d( q_x \delta r \wedge dt) + q_{xx}\delta r \wedge dt,\\
- r_x \delta q_x \wedge dx \wedge dt = d(r_x \delta q \wedge dt) + r_{xx} \delta q \wedge dt,
\end{gather}
we obtain 
\begin{multline}
\delta \Lambda = [(-i r_t + r_{xx} - 2   r^2 q) \delta q + (i q_t + q_{xx} - 2   r q^2)\delta r]\wedge dx \wedge dt\\
+ d(\frac{i}{2} r \delta q \wedge dx - \frac{i}{2} q \delta r \wedge dx + q_x \delta r \wedge dt + r_x \delta q \wedge dt)
\end{multline}
from which we can read off $\omegaone$.  We then compute $\Omega=\delta\omegaone$ to get the stated result.
\end{proof}
\begin{proposition}\label{decompositionNLS}
A Hamiltonian 1-form for the NLS equation is $F= F^1(q,r) dx + F^2(q,r, q_x, r_x) dt$, where
\begin{equation}
\parder{F^2}{r_x} =- i \parder{F^1}{r}, \qquad \parder{F^2}{q_x} =  i \parder{F^1}{q}.
\end{equation}
The respective Hamiltonian vector field is
\begin{equation}
X_F= \parder{F^2}{r_x}\partial_q + \parder{F^2}{q_x} \partial_r - \parder{F^2}{r} \partial_{q_x} - \parder{F^2}{q} \partial_{r_x}.
\end{equation}
Any two Hamiltonian 1-forms $F=Adx + Bdt$ and $G=Cdx + Ddt$ satisfy the equation
\begin{equation}
\cpb{F}{G} =  \{B,D\}_T dt-\{A,C\}_Sdx  
\end{equation} 
where the single-time Poisson Brackets are given by
\begin{equation}
 \{A, C \}_S =i\left( \parder{A}{q} \parder{C}{r}- \parder{C}{q} \parder{A}{r}\right)\,, \quad \{B, D \}_T = \parder{B}{q} \parder{D}{r_x} - \parder{D}{q} \parder{B}{r_x} + \parder{B}{r}\parder{D}{q_x} - \parder{D}{r} \parder{B}{q_x}\,.
\end{equation}
\end{proposition}
\begin{proof}
We start from the Ansatz $X_F = A \partial_q + B \partial_r + C \partial_{q_x} + D \partial_{r_x}$, and we want to find the coefficients by setting
\begin{equation}
\ip{X_F}{ \Omega} = \delta F.
\end{equation}
The right hand-side reads
\begin{multline}
\delta F = \parder{F^2}{q} \delta q \wedge dt + \parder{F^2}{r} \delta r \wedge dt + \parder{F^1}{q_x} \delta q_x \wedge dt + \parder{F^2}{r_x} \delta r_x \wedge dt\\
+ \parder{F^1}{q} \delta q \wedge dx + \parder{F^1}{r} \delta r \wedge dx,
\end{multline}
while the left hand-side is
\begin{equation}
\ip{X_F}{ \Omega} = i A \delta r \wedge dx + A \delta r_x \wedge dt - i B \delta q \wedge dx + B \delta q_x \wedge dt - C \delta r \wedge dt - D \delta q \wedge dt.
\end{equation}
By matching the coefficients we get
\begin{equation}
- D = \parder{F^2}{q}, \quad - C= \parder{F^2}{r}, \quad B = \parder{F^2}{q_x}, \quad A = \parder{F^2}{r_x}, \quad i B = -\parder{F^1}{q}, \quad iA= \parder{F^1}{r},
\end{equation}
which is the first statement. 
The second statement then follows by a direct calculation from $\cpb{F}{G}=- i(X_F) \delta G$ and recognizing the single-time Poisson brackets as defined in the Proposition.
\end{proof}
\begin{theorem}\label{covariantrmatrix:NLS}
The Lax form $W(\lda)=U(\lda)\,dx+V(\lda)\,dt$ satisfies the following covariant Poisson bracket 
\be
\cpb{W_1(\lambda)}{W_2(\mu)}=[r_{12}(\lda,\mu),W_1(\lambda)+W_2(\mu)]
\ee
where the classical $r$-matrix is that of the NLS equation (see e.g. \cite{FT}), the so-called rational $r$-matrix,
\be
\label{r_rational}
r_{12}(\lambda,\mu) = \frac{1}{\mu-\lambda}(\sigma_+\otimes \sigma_- + \sigma_-\otimes\sigma_+ + \sigma_3 \otimes \sigma_3 /2 + \id \otimes \id /2)\,.
\ee
\end{theorem}
\begin{proof} Again, we give here the proof by direct computation. 
We write $W_1(\lambda) = W^3(\lambda) \sigma_3 \otimes \id + W^+(\lambda) \sigma_+ \otimes \id + W^-(\lambda) \sigma_- \otimes \id$ and $W_2(\mu) = W^3(\mu) \id \otimes \sigma_3 + W^+(\mu) \id \otimes \sigma_+ + W^-(\mu) \id \otimes \sigma_-$. 
For the right-hand side, we find
\bea
\label{form_PB_NLS}
[r_{12}(\lambda-\mu), W_1(\lambda) + W_2(\mu)]&=&\frac{1}{\mu -\lambda} \left[(2W^3(\mu)-2W^3(\lambda))\sigma_+\otimes \sigma_- +(W^-(\lambda) - W^-(\mu))\sigma_3 \otimes \sigma_- \right. \nonumber\\
&&\left.+(W^+(\lambda) - W^+(\mu))\sigma_+\otimes \sigma_3+(2W^3(\lambda) - 2W^3(\mu))\sigma_- \otimes \sigma_+ \right. \nonumber\\
&&\left.+ (W^+(\mu)-W^+(\lambda))\sigma_3 \otimes \sigma_+ + (W^-(\mu)-W^-(\lambda))\sigma_-\otimes \sigma_3\right].\nonumber\\
&=& (-i dx -i(\lambda +\mu)dt)\,(\sigma_+\otimes \sigma_- - \sigma_- \otimes \sigma_+)+r\,dt\,( \sigma_- \otimes \sigma_3 -\sigma_3 \otimes \sigma_-)  \nonumber\\ 
&&+ q\, dt\,\sigma_+ \otimes \sigma_3 +q\, dt\,(\sigma_+\otimes \sigma_3-\sigma_3 \otimes \sigma_+ )
\eea
For the left-hand side, note that $W^3$, $W^+$ and $W^-$ are Hamiltonian forms. Thus, a direct calculation using the splitting formula shows that the only nonzero covariant Poisson bracket relations are the following
\bea
&&\cpb{W^+(\lambda)}{W^-(\mu)} =-i dx - i (\lambda + \mu) dt\,,\nonumber\\
&&\cpb{W^+(\lambda)}{W^3(\mu)} = -q dt\,,\nonumber\\
&&\cpb{W^-(\lambda)}{W^+(\mu)}= i dx + i (\lambda + \mu) dt\,,\nonumber\\
&&\cpb{W^-(\lambda)}{W^3(\mu)}= r dt\,,\nonumber\\
&&\cpb{W^3(\lambda)}{W^+(\mu)}=  q dt\,,\nonumber\\
&&\cpb{W^3(\lambda)}{W^-(\mu)} = -r dt\,.\nonumber
\eea
It remains to insert in the definition \eqref{def_matrix_cov_PB} to recognize that $\cpb{W_1(\lambda)}{W_2(\mu)}$ is precisely \eqref{form_PB_NLS}.
\end{proof}
We conlude the NLS example by a description of its covariant Hamiltonian formulation. We first compute the energy-momentum tensors $T_x= T_{xx} dt - T_{xt} dx$ and $T_t= T_{tx} dt - T_{tt} dx$ according to formula \eqref{EM_tensor} and find
$T_{xx} = \frac{i}{2} (q r_t - r q_t) - q_x r_x + q^2 r^2$ and $ T_{tt}= q_x r_x + q^2 r^2$. Hence, the covariant Hamiltonian $\mathcal{H} = H \vol$ is given by
\begin{equation}
H= q^2 r^2 - q_x r_x.
\end{equation}
The covariant Hamiltonian vector field $X_H$, such that $\ip{X_H}{ \Omega} = \delta H$ can be taken as
\begin{equation}
X_H = -(  i q^2r \partial_q - i q r^2 \partial_r ) \wedge \partial_x - (q_x \partial_q + r_x \partial_r + q^2 r \partial_{q_x} + qr^2 \partial_{r_x})\wedge \partial_t\,.
\end{equation}
Equipped with this, we have the following result.
\begin{proposition}
The covariant Hamiltonian formulation of the NLS equation is given by
\be
d W(\lambda)=\cpb{H}{W(\lambda)}\,\vol\,,
\ee
where $ W(\lambda)$ is the Lax Form.
\end{proposition}
\begin{proof} On the one hand
\begin{equation}
d W(\lambda) = ((- i q r_x - i r q_x) \sigma_3 + (- q_t + \lambda q_x + i q_{xx}) \sigma_+ + (- r_t + \lambda r_x - i r_{xx})\sigma_-)dx \wedge dt\,,
\end{equation}
while on the other hand,
\begin{equation}
\begin{split}
\cpb{H}{W(\lambda)} =& \ip{X_H }{\delta W(\lambda)} \\
=&\ip{X_H }{ (\sigma_+ \delta q \wedge dx + \sigma_- \delta r \wedge dx + (-i r \sigma_3 + \lambda \sigma_+) \delta q \wedge dt + (-i q \sigma_3 + \lambda \sigma_-) \delta r \wedge dt  \\ &+ i \sigma _+ \delta q_x \wedge dt - i \sigma_- \delta r_x \wedge dt)} \\
=&( 2i q^2 r + \lambda q_x) \sigma_+ + (-2i q r^2 + \lambda r_x) \sigma_- - (i q_x r + i qr_x) \sigma_3 \,.
\end{split}
\end{equation}
Therefore $d W(\lambda) = \cpb{H}{W(\lambda)} \vol $ reproduces the NLS equation.
\end{proof}
One can verify with direct computation that $\cpb{H}{W(\lambda)} = [U(\lambda), V(\lambda)]$.

\subsubsection{The modified KdV equation}

By a slight abuse of language, we call the following system of equations for two complex scalar fields $q,r$ the modified Korteweg-de Vries (KdV) equation,
\bea
\label{mKdV}
\begin{cases}
q_t - q_{xxx} + 6 qrq_x  =0\,,\\
r_t - r_{xxx} + 6 qrr_x =0 \,.
\end{cases}
\eea
It is the next commuting flow in the so-called AKNS hierarchy \cite{AKNS} that also contains the NLS system \eqref{NLS}. The original (real) modified KdV equation is obtained as the real reduction $r=q$ with $q$ a real-valued field. 
A Lagrangian form for \eqref{mKdV} is given by
\begin{equation}
\Lambda = [ \frac{i}{2}(r q_t - r_t q) +\frac{i}{2}(q_{xx} r_x - r_{xx} q_x ) - \frac{3i}{2}qr(qr_x - r q_x) ] dx \wedge dt\,.
\end{equation}
The system \eqref{mKdV} is equivalent to the zero curvature equation which must hold as an identity in $\lda$
\be
\partial_t U(\lda)-\partial_x V(\lda)+[U(\lda),V(\lda)]=0\,.
\ee
where the Lax pair $(U,V)$ can be taken as
\bea
U(\lambda)&=& -\frac{i \lambda}{2} \sigma_3 + q \sigma_+ + r \sigma_- \,,\\
V(\lambda)&=&(- \frac{\lambda^3}{2i} + i \lambda qr + r_x q - q_x r) \sigma_3 \nonumber\\
&&+ (-\lambda^2 q - i \lambda q_x + q_{xx} - 2 q^2 r) \sigma_+ + (- \lambda^2 r + i \lambda r_x + r_{xx} - 2 q r^2) \sigma_- \,.
\eea
In the general notations of Section \ref{generalities}, here $N=2$, $m=2$, and the two field are $u_1=q$ and $u_2=r$. We will denote $u_k^{(i)}$, $k=1,2$, $(i)=(0,0)$, $(1,0)$, etc. as $q$, $r$, $q_x$, $r_x$, etc. for convenience. One reason for looking at this model in addition to NLS, besides its physical relevance as a prototypical model related to the famous Korteweg-de Vries equation (by a Miura transformation \cite{Miura}), is that it is degenerate both in the standard Legendre transform and the dual one \cite{ACDK}. However, the method laid out by Dickey produces a multisymplectic form that is not sensitive to the degeneracy and both single-time forms are indeed symplectic (nondegenerate). In fact, they coincide with the ones obtained by the Dirac procedure in \cite{ACDK}. This feature is quite remarkable but its origin is not understood yet. As mentioned earlier, it might provide in the present multisymplectic context the analog of the argument popularised by Faddeev and Jackiw in \cite{FJ}.
\begin{proposition}
	The form $\omegaone$ is given by
	\be
	\omegaone=\frac{i}{2}\left[(q \delta r - r \delta q) \wedge dx + (2 q_{xx} \delta r - 2 r_{xx} \delta q + r_x \delta q_x - q_ x \delta r_x - 3 q^2 r \delta r + 3 qr^2 \delta q) \wedge dt\right]\,,
	\ee
	and the multisymplectic form reads
	\begin{equation}
	\Omega = i \delta q \wedge \delta r \wedge dx + ( i \delta q_{xx} \wedge \delta r - i \delta r_{xx} \wedge \delta q - 6 i qr \delta q \wedge \delta r + i \delta r_x \wedge \delta q_x) \wedge dt\,.
	\end{equation}
\end{proposition}
\begin{proof}
By direct calculation as in the two previous examples. 
\end{proof}

Note that for higher order field theories such as modified KdV, the form $\omegaone$ is not uniquely defined in general. However there is a property called quasisymmetry, introduced in \cite{Kolar}, that implies uniqueness of this form. In the present case, it boils down to checking that the coefficient of $\delta q_x\wedge dx$ in $\omegaone$, denoted $B_1^{12}$, equals the coefficient of $\delta q_t\wedge dt$, denoted $B_1^{21}$, and similarly for the coefficients of $\delta r_x\wedge dx$ and of $\delta r_t\wedge dt$, denoted $B_2^{12}$ and $B_2^{21}$ respectively. These coefficients are zero in our case so the property is satisfied.

\begin{proposition}\label{decompositionmKdV}
	A Hamiltonian 1-form for the mKdV equation is given by $F=  F^1(q,r) dx + F^2(q,r,q_x,r_x,q_{xx}, r_{xx}) dt$, where
	\begin{equation}
	\parder{F^2}{r_{xx}} =  \parder{F^1}{r} , \qquad \parder{F^2}{q_{xx}} =   \parder{F^1}{q} .
	\end{equation}
	The corresponding Hamiltonian vector field is
\bea
	X_F&=&- i \parder{F^1}{r} \partial_q +i \parder{F^1}{q}\partial_r + i \parder{F^2}{r_x} \partial_{q_x} - i \parder{F^2}{q_x} \partial_{r_x} \\
	&&- i(\parder{F^2}{r} + 6qr\parder{F^2}{r_{xx}}) \partial_{q_{xx}} + i (\parder{F^2}{q} + 6qr \parder{F^2}{q_{xx}})\partial_{r_{xx}}
\eea
		Any two Hamiltonian 1-forms $F=Adx + Bdt$ and $G=Cdx + Ddt$ satisfy the equation
	\begin{equation}
	\cpb{F}{G} =  \{B,D\}_T\, dt -\{A,C\}_S\,dx 
	\end{equation} 
	where the single-time Poisson Brackets are given by 
	\begin{equation}
	\{A, C \}_S =i\left( \parder{A}{q} \parder{C}{r}- \parder{C}{q} \parder{A}{r}\right),
	\end{equation}
\bea
	\{B, D \}_T &=&i\left( - \parder{B}{q} \parder{D}{r_{xx}} +  \parder{D}{q} \parder{B}{r_{xx}}  +  \parder{B}{r} \parder{D}{q_{xx}}- \parder{D}{r} \parder{B}{q_{xx}}\right) \nonumber\\
	&&+ i\left(\parder{B}{q_x} \parder{D}{r_x} - \parder{D}{q_x}\parder{B}{r_x} - 6qr\parder{B}{q_{xx}} \parder{D}{r_{xx}} + 6qr \parder{D}{q_{xx}} \parder{B}{r_{xx}} \right).
\eea
\end{proposition}
\begin{proof}
Inserting $X_F =A\partial_q + B \partial_r + C \partial_{q_x} + D \partial_{r_x} + E \partial_{q_{xx}} + G \partial_{r_{xx}}$ into
\begin{equation}
\ip{X_F}{ \Omega }= \delta F.
\end{equation}
We have to match the coefficients of 
\bea
\delta F &=& \parder{F^1}{q} \delta q \wedge dx + \parder{F^1}{r} \delta r \wedge dx 
+ \parder{F^2}{q} \delta q \wedge dt + \parder{F^2}{r} \delta r \wedge dt + \parder{F^2}{q_x} \delta q_x \wedge dt\nonumber\\ 
&&+ \parder{F^2}{r_x} \delta r_x \wedge dt + \parder{F^2}{q_{xx}} \delta q_{xx} \wedge dt + \parder{F^2}{r_{xx}} \delta r_{xx} \wedge dt
\eea
with those of
\bea
\ip{X_F}{\Omega} &=& i A \delta r \wedge dx + i A \delta r_{xx} \wedge dt - 6 i qrA \delta r \wedge dt - i B \delta q \wedge dx - i B \delta q_{xx} \wedge dt \nonumber\\
&& + 6iqrB \delta q \wedge dt - i C \delta r_x \wedge dt + i D \delta q_x \wedge dt + i E \delta r \wedge dt - i G \delta q \wedge dt.
\eea
This gives the first statement. The second statement then follows by a direct calculation from $\cpb{F}{G}=- \ip{X_F}{ \delta G}$ and recognizing the single-time Poisson brackets as defined in the Proposition.
\end{proof}

\begin{theorem}\label{covariantrmatrix:mKdV} 
	The Lax form $W(\lda)=U(\lda)\,dx+V(\lda)\,dt$ satisfies the following covariant Poisson bracket 
	\be
	\cpb{W_1(\lambda)}{W_2(\mu)}=[r_{12}(\lda,\mu),W_1(\lambda)+W_2(\mu)]
	\ee
	where $r$ is the rational classical $r$-matrix of the NLS equation. 
	\be
	r_{12}(\lambda-\mu) = \frac{1}{\mu-\lambda}(\sigma_+\otimes \sigma_- + \sigma_-\otimes\sigma_+ + \sigma_3 \otimes \sigma_3 /2 + \id \otimes \id /2)\,.
	\ee
\end{theorem}
\begin{proof}The direct calculation follows exactly the same steps as before so we only provide the main steps. 
On the one hand, we find 
\begin{multline}
\label{comm_mKdV}
	[r_{12}(\lambda-\mu), W_1(\lambda) + W_2(\mu)] = (-idx+i(\mu^2 +\mu \lambda + \lambda^2 +2qr)dt) (\sigma_+\otimes \sigma_- - \sigma_-\otimes \sigma_+) \\+ ((\lambda+ \mu)r-ir_x)dt (\sigma_3 \otimes \sigma_-  -  \sigma_-\otimes \sigma_3) + ((\mu +\lambda)q +iq_x)dt( \sigma_+\otimes \sigma_3-\sigma_3 \otimes \sigma_+) \,.
	\end{multline}
	On the other hand, writing $W = W^3 \sigma_3 + W^+\sigma_+ + W^- \sigma_-$ with
\begin{gather}
W^3(\lambda) = -\frac{i\lambda}{2} dx + ( -\frac{\lambda^3}{2i} + i \lambda qr + r_x q- q_x r)dt\\
W^+(\lambda) = q dx + (-\lambda^2 q - i \lambda q_x + q_{xx} -2 q^2r) dt\\
W^-(\lambda) = r dx + (-\lambda^2 r + i \lambda r_x + r_{xx} -2 qr^2) dt
\end{gather}
we find that these are Hamiltonian forms. Thus, the only nonzero covariant Poisson brackets are
\bea
&&	\cpb{W^+(\lambda)}{W^-(\mu)} =-i dx + i ( 2 qr + \mu^2 + \lambda \mu + \lambda^2) dt\,,\\
&&	\cpb{W^+(\lambda)}{W^3(\mu)} = ( q (\lambda+ \mu) + i q_x) dt\,,\\
&&	\cpb{W^-(\lambda)}{W^+(\mu)}= i dx - i(2qr + \mu^2 + \lambda \mu + \lambda^2)dt\,,\\
&&	\cpb{W^-(\lambda)}{W^3(\mu)}= (-r(\lambda + \mu) +i r_x) dt\,,\\
&&	\cpb{W^3(\lambda)}{W^+(\mu)}= - (q(\mu + \lambda) + i q_x) dt\,	,\\
&&	\cpb{W^3(\lambda)}{W^-(\mu)} = (r(\lambda + \mu) - i r_x) dt\,.
\eea
	which are combined according to \eqref{def_matrix_cov_PB} to find that $\cpb{W_1(\lambda)}{W_2(\mu)}$ is exactly equal to the right-hand side of \eqref{comm_mKdV}.
\end{proof}
A comment is in order regarding the fact that the same $r$-matrix as for the NLS appears here for the mKdV. In the standard Hamiltonian approach to the AKNS hierarchy, the only $r$-matrix structure is that given in \eqref{r_rational} since all the higher flows share the same $U$ matrix. In our covariant context, since the same $r$-matrix appears for both the $U$ and $V$ Lax matrices and since both flows share the same $U$, we consistently find that the same $r$-matrix appears in the covariant Poisson structure for NLS and mKdV. We note however that this points to a deeper connection between our covariant approach and the notion of integrable hierarchies. The study of such a connection is beyond the scope of the present paper and is left for future investigation. 

We conlude the mKdV example by a description of its covariant Hamiltonian formulation. We find the needed components of the energy-momentum tensors $T_x$ and $T_t$ to be $T_{xx} =  \frac{i}{2} (q r_t - r q_t) +i(q_{xx} r_x - r_{xx} q_x)$ and $T_{tt}= - \frac{i}{2} (q_{xx} r_x - r_{xx} q_x) + \frac{3i}{2} qr(q r_x - rq_x)$. Hence, the covariant Hamiltonian $\mathcal{H} = H \vol$ is given by
\begin{equation}
H= i(q_{xx} r_x - r_{xx} q_x)
\end{equation}
The covariant Hamiltonian vector field $X_H$, such that $\ip{X_H}{ \Omega} = \delta H$ can be taken as
\begin{equation}
X_H =-6 (qrq_x \partial_q + qrr_x \partial_r) \wedge \partial_x -(q_x \partial_q + r_x \partial_r + q_{xx} \partial_{q_x} + r_{xx} \partial_{r_x}) \wedge \partial_t.
\end{equation}
Equipped with this, we have the following result.
\begin{proposition}
	The covariant Hamiltonian formulation of the NLS equation is given by
	\be
	dW(\lambda) = \cpb{H}{W(\lambda)}\,\vol\,,
	\ee
	where $W(\lambda)$ is the Lax form.
\end{proposition}
\begin{proof} By direct computation as in the two previous examples.
\end{proof}
In the same way as in the two previous examples, one can show that $\cpb{H}{W(\lambda)} = [U(\lambda), V(\lambda)]$.

\section*{Conclusions}

By means of three of the most famous examples of integrable field theories, we established a systematic connection between the classical $r$-matrix formalism and the covariant Hamiltonian framework. The central result is the covariant Poisson structure of Theorem \ref{main_th} which represents the covariant analog of the celebrated (linear) Sklyanin Poisson algebra \cite{Skly,Skly_r} defined by the classical $r$-matrix. As this is the first time that such a connection is obtained, our results open the way to various investigations. We only mention a few here. 

From the point of view of the theory of the classical $r$-matrix, the examples considered here all belong to the so-called ultralocal case \ie the case of a skew-symmetric $r$-matrix. The non-ultralocal (non skew-symmetric) case is a natural open question, motivated by the fact that several important models of theoretical physics belong to this class \cite{Maillet}. 

Another natural question is that of the extension of the present work to an entire integrable hierarchy rather than a single representative of such a hierarchy in a two-dimensional spacetime. Indeed, the results of \cite{AC} show that the observation of the same $r$-matrix between any two pairs of times within an integrable hierarchy holds. We anticipate that a satisfactory answer to the question of a covariant formulation of the $r$-matrix structure of an entire hierarchy will rely on ideas put forward recently in \cite{SNC} regarding Lagrangian multiforms and a variational approach to Lax representations of integrable field theories. 
We hope that this could shed some light on the relationship between the present results and traditional features of integrable PDEs such as recursion operators.

It would also be desirable to achieve a fully fledged theory of the covariant $r$-matrix, combining the geometric, coordinate-independent formulation of variational calculus with the geometric and algebraic formulations of the $r$-matrix theory. This longer term goal is left for future work.

We note that the ideas of multisymplectic geometry have been very successfully applied to dispersive wave propagation problems \cite{Bridges} and to numerical integration algorithms \cite{BR}. However, it is not clear yet how our results could shed light on these frameworks.

\section*{Acknowledgements}
It is a pleasure to acknowledge helpful discussions with Raffaele Vitolo.

\end{document}